\documentclass[11pt,twoside,a4paper]{article} 
\usepackage[round,authoryear]{natbib}
\usepackage[colorlinks,linkcolor=blue,citecolor=blue,urlcolor=blue]{hyperref} 
\usepackage{color}
\usepackage{bm}
\usepackage{geometry}
\usepackage{amsthm,amsmath,amssymb}
\usepackage{mathrsfs}
\usepackage{enumerate}
\usepackage{amsfonts} 
\usepackage{bm}
\usepackage{CJK}
\usepackage{booktabs}
\usepackage {rotating}
\usepackage{threeparttable}
\usepackage{amsthm}
\usepackage{authblk} 

\def\ind{ {{\rm 1}\hskip-2.2pt{\rm l}}}
\begin{document}

\newcommand\keywords[1]{\textbf{Keywords}: #1} 
\date{} 

\allowdisplaybreaks[4]            
\renewcommand{\figurename}{Fig.}                                

\newcommand{\bA}{\bm A}
\newcommand{\ba}{\bm a}
\newcommand{\bB}{\bm B}
\newcommand{\bb}{\bm b}
\newcommand{\bC}{\bm C}
\newcommand{\bc}{\bm c}
\newcommand{\bD}{\bm D}
\newcommand{\bd}{\bm d}
\newcommand{\bE}{\bm E}
\newcommand{\be}{\bm e}
\newcommand{\bF}{\bm F}
\newcommand{\bmf}{\bm f}
\newcommand{\bG}{\bm G}
\newcommand{\bg}{\bm g}
\newcommand{\bH}{\bm H}
\newcommand{\bh}{\bm h}
\newcommand{\bI}{\bm I}
\newcommand{\bi}{\bm i}
\newcommand{\bJ}{\bm J}
\newcommand{\bj}{\bm j}
\newcommand{\bK}{\bm K}
\newcommand{\bk}{\bm k}
\newcommand{\bL}{\bm L}
\newcommand{\bl}{\bm l}
\newcommand{\bM}{\bm M}
\newcommand{\bmm}{\bm m}
\newcommand{\bN}{\bm N}
\newcommand{\bn}{\bm n}
\newcommand{\bO}{\bm O}
\newcommand{\bo}{\bm o}
\newcommand{\bP}{\bm P}
\newcommand{\bp}{\bm p}
\newcommand{\bQ}{\bm Q}
\newcommand{\bq}{\bm q}
\newcommand{\bR}{\bm R}
\newcommand{\br}{\bm r}
\newcommand{\bS}{\bm S}
\newcommand{\bs}{\bm s}
\newcommand{\bT}{\bm T}
\newcommand{\bt}{\bm t}
\newcommand{\bU}{\bm U}
\newcommand{\bu}{\bm u}
\newcommand{\bV}{\bm V}
\newcommand{\bv}{\bm v}
\newcommand{\bW}{\bm W}
\newcommand{\bw}{\bm w}
\newcommand{\bX}{\bm X}
\newcommand{\bx}{\bm x}
\newcommand{\bY}{\bm Y}
\newcommand{\by}{\bm y}
\newcommand{\bZ}{\bm Z}
\newcommand{\bz}{\bm z}

\newcommand{\balpha}{\bm \alpha}
\newcommand{\bbeta}{\bm \beta}
\newcommand{\bgamma}{\bm \gamma}
\newcommand{\bdelta}{\bm \delta}
\newcommand{\bepsilon}{\bm \epsilon}
\newcommand{\bvarepsilon}{\bm \varepsilon}
\newcommand{\bzeta}{\bm \zeta}
\newcommand{\bmeta}{\bm \eta}
\newcommand{\btheta}{\bm \theta}
\newcommand{\bvartheta}{\bm \vartheta}
\newcommand{\biota}{\bm \iota}
\newcommand{\bkappa}{\bm \kappa}
\newcommand{\blambda}{\bm \lambda}
\newcommand{\bmu}{\bm \mu}
\newcommand{\bnu}{\bm \nu}
\newcommand{\bxi}{\bm \xi}
\newcommand{\bmo}{\bm o}
\newcommand{\bpi}{\bm \pi}
\newcommand{\bvarpi}{\bm \varpi}
\newcommand{\brho}{\bm \rho}
\newcommand{\bvarrho}{\bm \varrho}
\newcommand{\bsigma}{\bm \sigma}
\newcommand{\bvarsigma}{\bm \varsigma}
\newcommand{\btau}{\bm \tau}
\newcommand{\bupsilon}{\bm \upsilon}
\newcommand{\bphi}{\bm \phi}
\newcommand{\bvarphi}{\bm \varphi}
\newcommand{\bchi}{\bm \chi}
\newcommand{\bpsi}{\bm \psi}
\newcommand{\bomega}{\bm \omega}
\newcommand{\bGamma}{\bm \Gamma}
\newcommand{\bDelta}{\bm \Delta}
\newcommand{\bTheta}{\bm \Theta}
\newcommand{\bLambda}{\bm \Lambda}
\newcommand{\bXi}{\bm \Xi}
\newcommand{\bPi}{\bm \Pi}
\newcommand{\bSigma}{\bm \Sigma}
\newcommand{\bUpsilon}{\bm \Upsilon}
\newcommand{\bPhi}{\bm \Phi}
\newcommand{\bPsi}{\bm \Psi}
\newcommand{\bOmega}{\bm \Omega}

\newtheorem{theorem}{\bf Theorem}
\newtheorem{lemm}{\bf Lemma}
\newtheorem{rem}{Remark}
\newtheorem{coll}{\bf Corollary}
\newtheorem{thm}{Theorem}
\newtheorem{lem}[thm]{Lemma}
\newtheorem{Proof}{proof}

\title{Testing for sufficient follow-up in censored survival data by using extremes}

\author[a,c]{Ping Xie}
\author[b]{Mikael Escobar-Bach}
\author[c]{Ingrid Van Keilegom \thanks{Corresponding author:ingrid.vankeilegom@kuleuven.be}}
\affil[a]{School of Mathematical Sciences, Dalian University of Technology, 116024 Dalian, Liaoning, China}
\affil[b]{Laboratoire Angevin de Recherche en Math\'{e}matiques, Universit\'{e} d'Angers, \\ 2 Boulevard de Lavoisier, 49035 Angers, France}
\affil[c]{Research Centre for Operations Research and Statistics, KU Leuven, \\ Naamsestraat 69, 3000 Leuven, Belgium}

\maketitle 

\begin{abstract}
In survival analysis, it often happens that some individuals, referred to as cured individuals, never experience the event of interest.
When analyzing time-to-event data with a cure fraction, it is crucial to check the assumption of `sufficient follow-up', which means that the right extreme of the censoring time distribution is larger than that of the survival time distribution for the non-cured individuals.
However, the available methods to test this assumption are limited in the literature. In this article, we study the problem of testing whether follow-up is sufficient for light-tailed distributions and develop a simple novel test. The proposed test statistic compares an estimator of the non-cure proportion under sufficient follow-up to one without the assumption of sufficient follow-up. A bootstrap procedure is employed to approximate the critical values of the test. We also carry out extensive simulations to evaluate the finite sample performance of the test and illustrate the practical use with applications to leukemia and breast cancer datasets.
\end{abstract}



\section{Introduction}\label{Intro}

In classical survival analysis, it is commonly assumed that all individuals in the study will eventually experience the event of interest if the follow-up period is long enough.
In reality, however, this assumption is often violated since some individuals will never experience the event of interest.
Here, those individuals are referred to as immune or cured individuals while the remaining individuals are often called susceptible.
The phenomenon of individuals being cured is quite common especially in cancer clinical studies, mainly thanks to the significant medical advancements that have been made.
For example, \cite{CA-SiegelEtal-2023} showed that the $5$-year relative survival rate for breast cancer is $91\%$ for patients diagnosed during 2012 through 2018.
Some patients with breast cancer will never experience a relapse or recurrence of breast cancer after receiving treatment and thus be cured.
We can also find other examples in various areas of applications, such as sociology (the time to become a parent), criminology (the time to commit crimes for a portion of the released population) and economics (the time to find a new job after a layoff), among others.

Cure models are frequently employed to analyze this type of data due to their ability to take the cure fraction into account.
The literature about cure models covers a wide range of research areas, including the estimation of the effects of covariates, the cure proportion and the survival function, and tests for the presence of cured individuals and for the validity of the assumed cure models, among others.
We refer the reader to \cite{Book_PengTaylor_2014}, \cite{ARSA_AmicoIngrid_2018}, and \cite{Book_Legrand_2019} for recent review papers on cure models, and to \cite{Book_MallerZhou_1996} and \cite{Book_PengYu_2021} for books on cure models.

When analyzing survival data with a cure fraction, it is essential to test the assumption of sufficient follow-up, which means that the right extreme of the censoring time distribution is larger than that of the survival time distribution for the susceptibles.
In practice, the length of the plateau in the Kaplan-Meier estimator (KME) \citep{JASA_KaplanMeier_1958} of the survival function and the number of censored observations contained in that plateau can serve as useful indicators to test the assumption of sufficient follow-up.
However, this judgement tends to be rather subjective and ambiguous.
To the best of our knowledge, there is still a scarcity of research on testing for sufficient follow-up in the statistical literature.
Based on the interval length from the largest uncensored observed survival time to the largest observed survival time, \cite{Biometrika_MallerZhou_1992} first proposed a testing procedure to this particular testing problem, and this test was studied in detail by \cite{JASA_MallerZhou_1994}.
However, their proposed test only estimates the significance level but does not control it, and the numerical study given in \cite{Book_MallerZhou_1996} showed that their test is too conservative.
Modifications for \cite{Biometrika_MallerZhou_1992}'s test were provided by \cite{Book_MallerZhou_1996} and \cite{SPL_Shen_2000}, respectively.
\cite{JSPI_KlebanovYakovlev_2007} discussed the aforementioned three testing procedures for the problem of testing for sufficient follow-up and found out that they fail to produce satisfactory results due to the highly unstable nature of the KME and related statistics in the presence of censoring.
Recently, \cite{CJS_MallerEtal_2023} filled in the gap in the asymptotic properties of the test statistic proposed by \cite{Book_MallerZhou_1996}, while the distributions of the tests proposed by \cite{Biometrika_MallerZhou_1992} and \cite{SPL_Shen_2000} remain to be investigated.
The simulation studies we conduct in Section~\ref{Simu_stud} reveal that the empirical levels of the test proposed by \cite{Book_MallerZhou_1996} are significantly higher than the nominal level especially when the follow-up period increases, and thus this testing procedure may incorrectly check the problem of testing for sufficient follow-up.

Whether follow-up is sufficient has a significant impact on the estimation of the cure proportion.
When the survival time is exposed to random right censoring, the non-cured individuals can observe either failure or censoring, whereas only censoring can be observed for all cured individuals.
Hence, a complicated problem arises regarding the identification of the cure rate due to the lack of data in the right tail of the survival function.
A solution to this problem is to impose the assumption of sufficient follow-up.
Under this assumption, \cite{Biometrika_MallerZhou_1992} suggested  a consistent nonparametric estimator of the cure rate that corresponds to the height of the plateau of the KME of the survival function.
In addition, \cite{CJS_XuPeng_2014} further proposed a nonparametric estimator of the cure rate in the presence of covariates and established its consistency and asymptotic normality.

When the assumption of sufficient follow-up is not satisfied, however, the nonparametric estimators proposed by \cite{Biometrika_MallerZhou_1992} and \cite{CJS_XuPeng_2014} will overestimate the cure rate, since the data do not provide enough information on the right tail of the survival function.
To tackle this problem, \cite{JRSSB_MikaelIngrid_2019} first considered extrapolation techniques from extreme value theory \citep{Book_DeFerreira_2006} to avoid the assumption of sufficient follow-up and corrected the estimator proposed by \cite{Biometrika_MallerZhou_1992}, which has also been extended to the case of the presence of covariates by \cite{CSDA_EscobaVanKeilegom_2023}.
The vital assumption imposed by \cite{JRSSB_MikaelIngrid_2019} is that the survival function of the non-cured individuals belongs to the Fr\'{e}chet domain of attraction.
\cite{Biometrika_Mikael_2022} further considered the case where the survival function of the non-cured individuals is in the Gumbel domain of attraction that covers most of the parametric distributions assumed in practice for analyzing survival data, including Weibull, log-normal, and Gamma distributions, among others.
Likewise, we assume that the distribution of the non-cured individuals belongs to the Gumbel domain of attraction in this article.

Extreme value theory is frequently utilized in the literature of survival analysis with right-censored data.
In addition to the aforementioned articles that applied extreme value theory to estimate the cure rate,
extreme value theory is also considered to estimate the extreme value index \citep[see, e.g.][]{SAJ_BeirlantGuillou_2001, Bernoulli_EinmahlFils_2008, Extremes_WormsWorms_2014,JMA_Stupfler_2016}.
Recently, \cite{Extremes_Maller_2022} developed the asymptotic distributions of statistics relevant to testing for the presence of cured individuals in the population through the application of extreme value theory.
\cite{Bernoulli_MallerEtal_2022} further studied the asymptotic distributions of the largest uncensored and largest observed survival times for survival data with a cure fraction.
Here, we will show that it becomes possible to construct a new test for the assumption of sufficient follow-up by using extreme value theory.

In this article, we propose a new test statistic for the null hypothesis of sufficient follow-up.
Instead of being based on the difference between the largest uncensored and the largest observed survival times, the proposed test statistic involves the comparison of two nonparametric estimators of the cure rate.
We reject the null hypothesis of sufficient follow-up if the value of the test statistic is large.
Moreover, the asymptotic distribution of our test statistic is shown to be a normal distribution under the null hypothesis.
Extensive simulation studies confirm that our proposed test performs satisfactorily for finite samples.
In particular, under insufficient follow-up, our test still works well in comparison to \cite{Book_MallerZhou_1996}'s testing procedure.

The article is organized as follows.
In Section~\ref{Test_proce}, we describe in detail the proposed testing procedure.
Section~\ref{Asym_theo} shows some asymptotic results, and a bootstrap procedure is also employed to approximate the critical values of the test in this section.
Simulation studies are reported in Section~\ref{Simu_stud} to evaluate the finite-sample performance of the test.
In Section~\ref{Real_data}, we illustrate the practical applicability through the analysis of both leukemia and breast cancer datasets.
Section~\ref{Diss} presents a discussion.
The detailed proofs of the main results are presented in the Appendix.

\section{Testing procedure}\label{Test_proce}

We consider a random right censorship model, and let $T$ be the survival time and $C$ be the random right censoring time.
The distribution functions of $T$ and $C$ are denoted as $F$ and $F_c$, respectively.
The population is divided into two groups consisting of the non-cured and cured individuals with proportions $p$ and $1-p$, $0<p<1$. The cured individuals, also called immune, never experience the event of interest and hence have by convention an infinite survival time. This leads to an improper distribution $F$ with total mass strictly less than 1 and the representation
\begin{eqnarray}
\label{eq::cure_model}
F(t)=\mathbb{P}(T \leq t)=pF_0(t),
\end{eqnarray}
where $F_0$ is a proper distribution function for the non-cured population, also called susceptibles. The value of $p$ actually represents the probability $\mathbb{P}(T<+\infty)$ while $1-p$ can be seen as the proportion of immunes, called the cure rate. The presence of right-censoring only allows us to observe the couple $(Y,\delta)$, where $Y=\min(T,C)$ and $\delta=\ind (T \leq C)$. Due to some identification issues, we impose the common assumption that $T$ and $C$ are independent. This particularly ensures that the distribution function $H$ of the observed random time $Y$ satisfies $1-H=(1-F)(1-F_c)$. In the sequel, we denote for convenience $\bar{G}=1-G$ and $\tau_G=\sup\{t\geq 0: G(t)<1\}$ as the right extreme for any generic distribution function $G$.

Let $\{(Y_i,\delta_i)\}_{1\leq i\leq n}$ be a $n$-sized random sample drawn from the aforementioned model, and $\tilde{t}_{(1)}<\tilde{t}_{(2)}<\ldots<\tilde{t}_{(K)}$ represent $K$ distinct uncensored survival times.
\cite{Biometrika_MallerZhou_1992} suggested an estimator of $p$ based on \eqref{eq::cure_model} by returning the maximum value of an estimate of the distribution $F$. Specifically, they proposed the following nonparametric estimator
\begin{eqnarray}\label{phat}
\hat{p}_n=\hat{F}_n(t_{(n)}),
\end{eqnarray}
where $t_{(n)}=\max\{Y_i,i=1,\ldots,n\}$ is the largest observed time in the sample, $\hat{F}_n(t)$ is the KME of $F$ defined as
\begin{eqnarray*}
\hat{F}_n(t)=1-\prod_{\tilde{t}_{(i)}\leq t}\left(1-\frac{d_i}{D_i}\right),
\end{eqnarray*}
$D_i$ is the number of individuals at risk at $\tilde{t}_{(i)}$, and $d_i$ is the number of individuals having failure at $\tilde{t}_{(i)}$.
As discussed in \cite{Biometrika_MallerZhou_1992}, $\hat{p}_n$ is consistent if and only if the assumption of sufficient follow-up is valid, which is equivalent in this case to $\tau_{F_{0}} \leq \tau_{F_c}$.
The case of insufficient follow-up however, when $\tau_{F_{0}} > \tau_{F_c}$, yields an underestimation of $p$ by $\hat{p}_n$.

Recently, \cite{Biometrika_Mikael_2022} proposed a nonparametric estimator $\hat{p}_G(\epsilon)$ of $p$ when $\bar{F}_0$ is in a class of particular slowly varying functions, where the definition of $\epsilon$ is given below.
More precisely, they assumed that the susceptible distribution is in the Gumbel domain of attraction.
This condition invokes the tail behaviour of the distribution and considers that there exists a positive function $l$ such that for all $y\in\mathbb{R}$,
\begin{eqnarray}
\lim_{t \rightarrow \tau_{F_0}}\dfrac{\bar{F}_0(t+yl(t))}{\bar{F}_0(t)}=e^{-y}.
\end{eqnarray}
In their paper, they also imposed a condition, namely the Von Mises condition \citep[p.42, Eq. 1.8]{Book_Resnick_1987}, for the distribution $F_0$.
This condition is a sufficient but not necessary condition for a distribution belonging to a domain of attraction.
For the Gumbel domain of attraction, the Von Mises condition says that $F_0$ has a finite second derivative $F''_0$ in a neighbourhood $(\tau_{F_0}-\epsilon, \tau_{F_0})$ for some $\epsilon \in (0,\tau_{F_0})$ with
\begin{eqnarray} \label{vonmisescondition}
\lim_{t\rightarrow \tau_{F_0}} \frac{F_0''(t)\bar{F}_0(t)}{\{F_0'(t)\}^2}=-1.
\end{eqnarray}
In practice, most of the distributions in the Gumbel domain satisfy the Von Mises condition, such as Weibull, log-normal and Gamma distributions.
The nonparametric estimator is then given by
\begin{eqnarray} \label{pGhat}
\hat{p}_G(\epsilon)=\hat{F}_n(t_{(n)}-\epsilon)+\frac{\{\hat{F}_n(t_{(n)}-\epsilon/2)-\hat{F}_n(t_{(n)}-\epsilon)\}^2}
 {2 \hat{F}_n(t_{(n)}-\epsilon/2)-\hat{F}_n(t_{(n)}-\epsilon)-\hat{F}_n(t_{(n)})}.
\end{eqnarray}
Under some assumptions, \cite{Biometrika_Mikael_2022} showed that the latter estimator $\hat{p}_G(\epsilon)$ converges in probability to
\begin{eqnarray}\label{PG}
p_G(\epsilon)=F(\tau_{H}-\epsilon)+\frac{\{F(\tau_{H}-\epsilon/2)-F(\tau_{H}-\epsilon)\}^2}
 {2 F(\tau_{H}-\epsilon/2)-F(\tau_{H}-\epsilon)-F(\tau_{H})}
\end{eqnarray}
as $n \rightarrow \infty$, and at the same time, converges to $p$ as $\epsilon \rightarrow 0$ and $\tau_{F_c} \rightarrow \tau_{F_0}$.
Roughly speaking, the idea relies on correcting the estimator $\hat{p}_n$ by some additive statistics based on the extreme value condition, compensating the underestimation of $\hat{p}_n$ under insufficient follow-up. Reversely, this latter correction turns out to be almost null under sufficient follow-up.

In this article, we are interested in testing the assumption of sufficient follow-up.
The hypotheses are
\begin{eqnarray*}
H_0: \tau_{F_0} \leq \tau_{F_c}\quad \text{versus}\quad H_1: \tau_{F_0} > \tau_{F_c}.
\end{eqnarray*}
The above arguments motivate us to propose a test statistic through the difference between $\hat{p}_n$ and $\hat{p}_G(\epsilon)$. Specifically, we define the test statistic as
\begin{eqnarray*}
T_n=\hat{p}_G(\epsilon)-\hat{p}_n.
\end{eqnarray*}
Under sufficient follow-up, $T_n$ will be close to zero while under insufficient follow-up, $\hat{p}_G(\epsilon)$ will be strictly larger than $\hat{p}_n$ and so $T_n$ will be strictly positive.
For large values of $T_n$, we hence expect to reject the null hypothesis $H_0$. To determine the rejection threshold for the null hypothesis, we will establish in the next section the asymptotic normality of $T_n$ under $H_0$.

We mention that under insufficient follow-up, $\epsilon$ belongs to the interval $[0,\tau_{F_c}]$.
Under sufficient follow-up, the assumption that $\tau_{F_c}<\infty$ is necessary to ensure that the denominator of \eqref{PG} is not equal to $0$, and it is easy to prove that $\hat{p}_G(\epsilon)$ converges in probability to~\eqref{PG} as $n \rightarrow \infty$.
However, it is worthwhile to point out that when $\tau_{F_0} \leq \tau_{F_c}$ and $\epsilon \in [0, \tau_{F_c}-\tau_{F_0}]$, the denominator of \eqref{PG} equals $0$. To avoid this case, we impose that $\epsilon$ belongs to the interval $\Theta=(0\vee (\tau_{F_c}-\tau_{F_0}), \tau_{F_c}]$. In practice, $\epsilon$ can be chosen in the interval $(t_{(n)}-\tilde{t}_{(K)}, t_{(n)})$, where $\tilde{t}_{(K)}$ is the largest uncensored observed survival time.

\section{Asymptotic theory}\label{Asym_theo}

In this section, we will establish the asymptotic normality of the test statistic $T_n$. With this in mind, we state the following necessary assumptions.
\begin{enumerate}
\item[1.] $F$ is continuous at $\tau_{F_0}$ in case $\tau_{F_0}<\infty$.
\item[2.] $\int_{[0,\tau_{F_0}]}d F_0(s)/\{1-F_c(s^{-})\}< \infty$.
\item[3.] $\lim_{n\rightarrow \infty}n \bar{F}_c(\tau_{F_c}-\xi/\sqrt{n})=\infty$ for each $\xi> 0$.
\end{enumerate}

In practice, Assumption~1 is valid for a wide class of distributions and Assumption~2 is a sufficient condition for $\tau_{F_0} \leq \tau_{F_c}$.
The limit $\lim_{n\rightarrow \infty}n \bar{F}_c(\tau_{F_c}-\xi/\sqrt{n})=\infty$ in Assumption~3 holds if the censoring distribution $F_c$ possesses a finite right extreme value $\tau_{F_c}$ and has a positive derivative in a left neighborhood of $\tau_{F_c}$.
This condition is easily achievable in real-world situations.
For instance, the uniform distribution on $[0,\tau]$ has a positive derivative of $1/\tau$ for any $\tau<\infty$ and thus satisfies this condition.

The asymptotic normality of the KME of $F$ \citep[Theorem 3.14]{Book_MallerZhou_1996} plays an important role in proving the asymptotic normality of $T_n$.
To do so, we first introduce some notations. Suppose $Z(t)$ is a Gaussian process on $[0,\tau_H]$ with zero mean and variance function
\begin{eqnarray*}
v(t)=\int_{[0,t]}\frac{d F(s)}{\bar{F}(s)\{1-F(s^{-})\}\{1-F_c(s^{-})\}}.
\end{eqnarray*}
We also define $D[0,\tau_H]$ as the functional space of c\`{a}dl\`{a}g real functions on $[0,\tau_H]$ associated with Skorokhod's topology. The asymptotic normality of the KME of $F$ is summarized in the following theorem.
\begin{theorem}[Theorem 3.14 from \cite{Book_MallerZhou_1996}]\label{theorem1}
Suppose that $F$ is continuous at $\tau_H$ in case $\tau_H< \infty$.
Suppose furthermore that $\lim_{t\rightarrow \tau_H} \{F(\tau_H)-F(t)\}^2 v(t)=0$ and
\begin{eqnarray}\label{eq::theorem1:1}
\lim_{t\rightarrow \tau_H} \int_{(t,\tau_H)} \frac{I\{0 \leq F_c(s^{-}) <1\}\bar{F}(s)}{\{1-F_c(s^{-})\}\{1-F(s^{-})\}}d F(s)=0.
\end{eqnarray}
Then $\lim_{t\rightarrow \tau_H} \bar{F}(t)Z(t)$ exists and is finite almost surely, and for each $t \geq 0$,
\begin{eqnarray}\label{Xn_expression}
X_n(t)=\sqrt{n}\frac{\bar{F}(t)}{\bar{F}(t\wedge t_{(n)})} \{\hat{F}_n(t\wedge t_{(n)})-F(t\wedge t_{(n)})\}
\end{eqnarray}
converges weakly in $D[0,\tau_H]$ to $X(t)=I(t \in [0, \infty))\bar{F}(t)Z(t)+I(t=\infty)R$ as $n \rightarrow \infty$,
where
\begin{eqnarray}
R=\left\{
\begin{aligned}
&\bar{F}(\tau_H)Z(\tau_H)\quad \text{if}\quad H(\tau_{H^{-}})<1, \\
&\lim_{t\rightarrow \tau_H} \bar{F}(t)Z(t) \quad \text{if}\quad H(\tau_{H^{-}})=1.
\end{aligned}
\right.
\end{eqnarray}
\end{theorem}
The following theorem summarizes the asymptotic normality of $T_n$ under sufficient follow-up.
\begin{theorem}\label{theorem2}
Suppose that Assumptions~1-3 hold and $F_0$ satisfies \eqref{vonmisescondition}.
Suppose also that $\tau_{F_c} < 2 \tau_{F_0}$.
Then for each $\epsilon \in (2(\tau_{F_c}-\tau_{F_0}), \tau_{F_c})$, the test statistic $T_n$ is asymptotically normal, $\sqrt{n}(T_n -b_{\epsilon})\rightarrow N(0, \sigma^2_{\epsilon})$ in distribution as $n \rightarrow \infty$, where the asymptotic bias is
\begin{eqnarray*}
b_{\epsilon}=F(\tau_{H}-\epsilon)+\frac{\{F(\tau_{H}-\epsilon/2)-F(\tau_{H}-\epsilon)\}^2}
 {2 F(\tau_{H}-\epsilon/2)-F(\tau_{H}-\epsilon)-F(\tau_{H})}-F(\tau_{H}),
\end{eqnarray*}
the asymptotic variance is
\begin{eqnarray*}
\sigma^2_{\epsilon}=\sum_{i=1}^3\sum_{j=1}^3 s_i s_j \bar{F}\{\tau_H+\epsilon(i-3)/2\}
\bar{F}\{\tau_H+\epsilon(j-3)/2\} v\{\tau_H+\epsilon(i\wedge j-3)/2\},
\end{eqnarray*}
and $s_1$, $s_2$, and $s_3$ are provided in the Appendix as \eqref{s1}, \eqref{s2} and \eqref{s3}, respectively.
\end{theorem}

The asymptotic normality of $T_n$ is based on Theorem~\ref{theorem1}.
In Theorem~\ref{theorem2}, we impose the condition that $\tau_{F_c} < 2 \tau_{F_0}$.
This condition can also be found in Theorem~2 in \cite{CJS_MallerEtal_2023} and guarantees that $2(\tau_{F_c}-\tau_{F_0})$ is strictly less than $\tau_{F_c}$.
When $\epsilon$ falls within the interval $(\tau_{F_c}-\tau_{F_0}, 2(\tau_{F_c}-\tau_{F_0})]$, the asymptotic distribution of $T_n$ is degenerate, and thus the range of values for $\epsilon$ in Theorem~\ref{theorem2} is limited to the interval $(2(\tau_{F_c}-\tau_{F_0}), \tau_{F_c})$.
It is worth mentioning that $F(\tau_{H})$ is equal to $p$ since $\tau_{H}=\tau_{F_c}$ and $\tau_{F_0} \leq \tau_{F_c}$.
Besides, we can see that the asymptotic bias $b_{\epsilon}$ converges to $0$ under $H_0$ as $\epsilon$ goes to $2(\tau_{F_c}-\tau_{F_0})$.
The proof of Theorem \ref{theorem2} is outlined in the Appendix.

Through the preceding discussion, we successfully established the asymptotic normality of the test statistic.
However, for small sample sizes, the normal approximation given in Theorem~\ref{theorem2} does not perform well.
Therefore, a bootstrap procedure is used to approximate the critical values of the test.
The proposed bootstrap procedure comprises the following sequential steps.

\begin{enumerate}
\item[(1)] Use the original sample to calculate the test statistic $T_n$.
\item[(2)] Draw $B$ naive bootstrap samples with replacement from the original sample.
\item[(3)] For $b=1, \ldots , B$, calculate the bootstrap test statistic $T_n^b$ by using the $b$-th naive bootstrap sample.
\item[(4)] The critical value of the test is approximated by the $\{(1-\alpha)B\}$th quantile of $\{T_n^1-T_n,\ldots,T_n^B-T_n\}$, where $\alpha$ is the level of the test.
\end{enumerate}

\section{Simulation studies}\label{Simu_stud}

In the following simulation studies, we consider three distributions for the failure times of the susceptible patients:
\begin{itemize}
\item[(i)] a Weibull distribution with shape parameter $1.5$ and scale parameter $1.5$;
\item[(ii)] an exponential distribution with parameter 1;
\item[(iii)] a standard log-normal distribution.
\end{itemize}
In addition, we generate the censoring variable from a Weibull distribution with shape parameter 1 and scale parameter $\lambda = 1.5, 2, 2.5, 3$ under the null hypothesis, and from uniform distributions $U[0,\mu]$ with $\mu =2, 2.5, 3, 3.5$ under the alternative hypothesis.
It is readily apparent that smaller values of $\lambda$ correspond to more censoring.
Similarly, smaller values of $\mu$ mean a shorter follow-up period while large values represent light-censoring.
Two values for the non-cure rate have been chosen with $p$ equal to $0.7$ or $0.9$, representing moderate and high non-cure rates. Four sample sizes are considered with $n=400,800,1200,1800$.
The nominal level equals $\alpha=0.05$ and all simulation results are based on $1000$ simulation runs.

\begin{table}
 \centering
 \def\~{\hphantom{0}}
  \caption{Simulation results of our test, when the failure time of susceptible patients follows a Weibull distribution with shape parameter 1.5 and scale parameter 1.5 \label{Tab1}}
       \small
  \begin{tabular*}{\textwidth}{@{} *{15}{c@{\extracolsep{\fill}}}}
  \hline
  & & \multicolumn{2}{c}{$\lambda=3$} & \multicolumn{2}{c}{$\lambda=2.5$} &\multicolumn{2}{c}{$\lambda=2$} &\multicolumn{2}{c}{$\lambda=1.5$}\\ [1pt]
     \cline{3-4} \cline{5-6}  \cline{7-8} \cline{9-10}\\ [-6pt]
& & $p=0.9$ & $p=0.7$ & $p=0.9$ & $p=0.7$ & $p=0.9$ & $p=0.7$ & $p=0.9$ & $p=0.7$ \\
& & $27\%$ cens & $43\%$ cens & $30\%$ cens & $45\%$ cens & $33\%$ cens & $48\%$ cens & $39\%$ cens & $52\%$ cens \\
& $n$ & $10\%$ cure & $30\%$ cure & $10\%$ cure & $30\%$ cure & $10\%$ cure & $30\%$ cure & $10\%$ cure & $30\%$ cure \\
$H_0$ & ~400 & 0.036  & 0.067 & 0.041  & 0.056 & 0.049  & 0.043 & 0.106  & 0.038  \\
      & ~800 & 0.052  & 0.053 & 0.043  & 0.053 & 0.033  & 0.042 & 0.061  & 0.031  \\
      & 1200 & 0.058  & 0.047 & 0.055  & 0.057 & 0.045  & 0.044 & 0.035  & 0.031  \\
      & 1800 & 0.043  & 0.047 & 0.058  & 0.047 & 0.038  & 0.048 & 0.038  & 0.038  \\
       \hline
& & \multicolumn{2}{c}{$\mu=3.5$} & \multicolumn{2}{c}{$\mu=3$} &\multicolumn{2}{c}{$\mu=2.5$} &\multicolumn{2}{c}{$\mu=2$}\\
     \cline{3-4} \cline{5-6}  \cline{7-8} \cline{9-10}\\ [-6pt]
& & $p=0.9$ & $p=0.7$ & $p=0.9$ & $p=0.7$ & $p=0.9$ & $p=0.7$ & $p=0.9$ & $p=0.7$ \\
& & $28\%$ cens & $44\%$ cens & $31\%$ cens & $46\%$ cens & $34\%$ cens & $49\%$ cens & $40\%$ cens & $53\%$ cens \\
&$n$ & $10\%$ cure & $30\%$ cure & $10\%$ cure & $30\%$ cure & $10\%$ cure & $30\%$ cure & $10\%$ cure & $30\%$ cure \\
$H_1$ & ~400 & 0.711  & 0.134& 0.860  & 0.361& 0.922  & 0.640& 0.852  & 0.652  \\
      & ~800 & 0.930  & 0.400& 0.971  & 0.538& 0.973  & 0.725& 0.930  & 0.811  \\
      & 1200 & 0.990  & 0.812& 0.992  & 0.812& 0.990  & 0.891& 0.971  & 0.908  \\
      & 1800 & 0.996  & 0.939& 0.996  & 0.957& 0.994  & 0.985& 0.974  & 0.942  \\
\hline
\end{tabular*}
\bigskip \smallskip \smallskip
  \begin{tablenotes}
  \item NOTE: $H_0$, the null hypothesis; $H_1$, the alternative hypothesis; $n$, the sample size; $p$, the non-cure rate; $\lambda$, the parameter of the censoring distribution under $H_0$; $\mu$, the parameter of the censoring distribution under $H_1$.
  \end{tablenotes}
\end{table}

\begin{table}
 \centering
 \def\~{\hphantom{0}}
  \caption{Simulation results of our test, when the failure time of susceptible patients follows an exponential distribution with parameter 1 \label{Tab2}}
       \small
  \begin{tabular*}{\textwidth}{@{} *{15}{c@{\extracolsep{\fill}}}}
  \hline
& & \multicolumn{2}{c}{$\lambda=3$} & \multicolumn{2}{c}{$\lambda=2.5$} &\multicolumn{2}{c}{$\lambda=2$} &\multicolumn{2}{c}{$\lambda=1.5$}\\
     \cline{3-4} \cline{5-6}  \cline{7-8} \cline{9-10}\\ [-6pt]
& & $p=0.9$ & $p=0.7$ & $p=0.9$ & $p=0.7$ & $p=0.9$ & $p=0.7$ & $p=0.9$ & $p=0.7$ \\
& & $33\%$ cens & $47\%$ cens & $36\%$ cens & $50\%$ cens & $40\%$ cens & $53\%$ cens & $46\%$ cens & $58\%$ cens \\
 & $n$ & $10\%$ cure & $30\%$ cure & $10\%$ cure & $30\%$ cure & $10\%$ cure & $30\%$ cure & $10\%$ cure & $30\%$ cure \\
$H_0$ & ~400 & 0.053  & 0.049 & 0.038  & 0.043 & 0.055  & 0.049 & 0.086  & 0.035  \\
      & ~800 & 0.043  & 0.057 & 0.046  & 0.054 & 0.036  & 0.044 & 0.045  & 0.023  \\
      & 1200 & 0.054  & 0.047 & 0.042  & 0.051 & 0.046  & 0.047 & 0.042  & 0.038  \\
      & 1800 & 0.066  & 0.046 & 0.039  & 0.040 & 0.040  & 0.042 & 0.032  & 0.035  \\
       \hline
& & \multicolumn{2}{c}{$\mu=3.5$} & \multicolumn{2}{c}{$\mu=3$} &\multicolumn{2}{c}{$\mu=2.5$} &\multicolumn{2}{c}{$\mu=2$}\\
     \cline{3-4} \cline{5-6}  \cline{7-8} \cline{9-10}\\ [-6pt]
& & $p=0.9$ & $p=0.7$ & $p=0.9$ & $p=0.7$ & $p=0.9$ & $p=0.7$ & $p=0.9$ & $p=0.7$ \\
& & $35\%$ cens & $50\%$ cens & $38\%$ cens & $52\%$ cens & $43\%$ cens & $56\%$ cens & $49\%$ cens & $60\%$ cens \\
& $n$& $10\%$ cure & $30\%$ cure & $10\%$ cure & $30\%$ cure & $10\%$ cure & $30\%$ cure & $10\%$ cure & $30\%$ cure \\
$H_1$ & ~400 & 0.266  & 0.026 & 0.341  & 0.051 & 0.486  & 0.108 & 0.536  & 0.197  \\
      & ~800 & 0.496  & 0.075 & 0.624  & 0.111 & 0.705  & 0.139 & 0.747  & 0.178  \\
      & 1200 & 0.817  & 0.443 & 0.898  & 0.476 & 0.938  & 0.467 & 0.942  & 0.392  \\
      & 1800 & 0.956  & 0.776 & 0.981  & 0.789 & 0.992  & 0.759 & 0.995  & 0.695  \\
\hline
\end{tabular*}
\bigskip \smallskip \smallskip
  \begin{tablenotes}
  \item NOTE: $H_0$, the null hypothesis; $H_1$, the alternative hypothesis; $n$, the sample size; $p$, the non-cure rate; $\lambda$, the parameter of the censoring distribution under $H_0$; $\mu$, the parameter of the censoring distribution under $H_1$.
  \end{tablenotes}
\end{table}

\begin{table}
 \centering
 \def\~{\hphantom{0}}
  \caption{Simulation results of our test, when the failure time of susceptible patients follows a standard log-normal distribution \label{Tab3}}
       \small
  \begin{tabular*}{\textwidth}{@{} *{15}{c@{\extracolsep{\fill}}}}
  \hline
& & \multicolumn{2}{c}{$\lambda=3$} & \multicolumn{2}{c}{$\lambda=2.5$} &\multicolumn{2}{c}{$\lambda=2$} &\multicolumn{2}{c}{$\lambda=1.5$}\\
     \cline{3-4} \cline{5-6}  \cline{7-8} \cline{9-10}\\ [-6pt]
& & $p=0.9$ & $p=0.7$ & $p=0.9$ & $p=0.7$ & $p=0.9$ & $p=0.7$ & $p=0.9$ & $p=0.7$ \\
& & $41\%$ cens & $54\%$ cens & $44\%$ cens & $57\%$ cens & $50\%$ cens & $61\%$ cens & $56\%$ cens & $66\%$ cens \\
 & $n$ & $10\%$ cure & $30\%$ cure & $10\%$ cure & $30\%$ cure & $10\%$ cure & $30\%$ cure & $10\%$ cure & $30\%$ cure \\
$H_0$ & ~400 & 0.069  & 0.036& 0.056  & 0.031  & 0.089  & 0.033  & 0.138  & 0.056  \\
      & ~800 & 0.043  & 0.037& 0.047  & 0.033  & 0.073  & 0.034  & 0.125  & 0.029  \\
      & 1200 & 0.038  & 0.034& 0.040  & 0.027  & 0.055  & 0.029  & 0.086  & 0.029  \\
      & 1800 & 0.032  & 0.018& 0.044  & 0.031  & 0.050  & 0.021  & 0.096  & 0.021  \\
       \hline
& & \multicolumn{2}{c}{$\mu=3.5$} & \multicolumn{2}{c}{$\mu=3$} &\multicolumn{2}{c}{$\mu=2.5$} &\multicolumn{2}{c}{$\mu=2$}\\
     \cline{3-4} \cline{5-6}  \cline{7-8} \cline{9-10}\\ [-6pt]
& & $p=0.9$ & $p=0.7$ & $p=0.9$ & $p=0.7$ & $p=0.9$ & $p=0.7$ & $p=0.9$ & $p=0.7$ \\
& & $45\%$ cens & $57\%$ cens & $49\%$ cens & $60\%$ cens & $54\%$ cens & $64\%$ cens & $60\%$ cens & $69\%$ cens \\
&$n$ & $10\%$ cure & $30\%$ cure & $10\%$ cure & $30\%$ cure & $10\%$ cure & $30\%$ cure & $10\%$ cure & $30\%$ cure \\
$H_1$ & ~400 & 0.258  & 0.047 & 0.310  & 0.094 & 0.384  & 0.129& 0.483  & 0.218  \\
      & ~800 & 0.544  & 0.108 & 0.590  & 0.114 & 0.620  & 0.124& 0.647  & 0.196  \\
      & 1200 & 0.885  & 0.509 & 0.878  & 0.482 & 0.899  & 0.444& 0.884  & 0.327  \\
      & 1800 & 0.966  & 0.790 & 0.972  & 0.795 & 0.978  & 0.768& 0.977  & 0.676  \\
\hline
\end{tabular*}
\bigskip \smallskip \smallskip
  \begin{tablenotes}
  \item NOTE: $H_0$, the null hypothesis; $H_1$, the alternative hypothesis; $n$, the sample size; $p$, the non-cure rate; $\lambda$, the parameter of the censoring distribution under $H_0$; $\mu$, the parameter of the censoring distribution under $H_1$.
  \end{tablenotes}
\end{table}

\begin{table}
 \centering
 \def\~{\hphantom{0}}
  \caption{Simulation results of the test proposed by Maller and Zhou (1996), when the failure time of susceptible patients follows a Weibull distribution with shape parameter 1.5 and scale parameter 1.5 \label{Tab4}}
       \small
  \begin{tabular*}{\textwidth}{@{} *{15}{c@{\extracolsep{\fill}}}}
  \hline
& & \multicolumn{2}{c}{$\mu=3.5$} & \multicolumn{2}{c}{$\mu=3$} &\multicolumn{2}{c}{$\mu=2.5$} &\multicolumn{2}{c}{$\mu=2$}\\
     \cline{3-4} \cline{5-6}  \cline{7-8} \cline{9-10}\\ [-6pt]
& & $p=0.9$ & $p=0.7$ & $p=0.9$ & $p=0.7$ & $p=0.9$ & $p=0.7$ & $p=0.9$ & $p=0.7$ \\
& & $28\%$ cens & $44\%$ cens & $31\%$ cens & $46\%$ cens & $34\%$ cens & $49\%$ cens & $40\%$ cens & $53\%$ cens \\
& $n$& $10\%$ cure & $30\%$ cure & $10\%$ cure & $30\%$ cure & $10\%$ cure & $30\%$ cure & $10\%$ cure & $30\%$ cure \\
$\tilde{H}_0$ & ~400& 0.187  & 0.214  & 0.146  & 0.172  & 0.113  & 0.133  & 0.114  & 0.094  \\
          & ~800& 0.181  & 0.177  & 0.138  & 0.148  & 0.102  & 0.137  & 0.088  & 0.097  \\
          & 1200& 0.152  & 0.184  & 0.137  & 0.120  & 0.109  & 0.115  & 0.090  & 0.107  \\
          & 1800& 0.154  & 0.162  & 0.129  & 0.130  & 0.098  & 0.099  & 0.085  & 0.088  \\
           \hline
& & \multicolumn{2}{c}{$\lambda=3$} & \multicolumn{2}{c}{$\lambda=2.5$} &\multicolumn{2}{c}{$\lambda=2$} &\multicolumn{2}{c}{$\lambda=1.5$}\\
     \cline{3-4} \cline{5-6}  \cline{7-8} \cline{9-10}\\ [-6pt]
& & $p=0.9$ & $p=0.7$ & $p=0.9$ & $p=0.7$ & $p=0.9$ & $p=0.7$ & $p=0.9$ & $p=0.7$ \\
& & $27\%$ cens & $43\%$ cens & $30\%$ cens & $45\%$ cens & $33\%$ cens & $48\%$ cens & $39\%$ cens & $52\%$ cens \\
 & $n$ & $10\%$ cure & $30\%$ cure & $10\%$ cure & $30\%$ cure & $10\%$ cure & $30\%$ cure & $10\%$ cure & $30\%$ cure \\
$\tilde{H}_1$& ~400 & 0.985  & 0.999  & 0.965  & 0.997  & 0.907  & 0.987  & 0.789  & 0.945  \\
         & ~800 & 1.000  & 1.000  & 0.992  & 0.999  & 0.980  & 0.997  & 0.853  & 0.970  \\
         & 1200 & 1.000  & 1.000  & 0.999  & 1.000  & 0.981  & 0.998  & 0.916  & 0.987  \\
         & 1800 & 1.000  & 1.000  & 0.998  & 1.000  & 0.990  & 0.998  & 0.940  & 0.996  \\
\hline
\end{tabular*}
\bigskip \smallskip \smallskip
  \begin{tablenotes}
  \item NOTE: $\tilde{H}_0$, the null hypothesis; $\tilde{H}_1$, the alternative hypothesis; $n$, the sample size; $p$, the non-cure rate; $\mu$, the parameter of the censoring distribution under $\tilde{H}_0$; $\lambda$, the parameter of the censoring distribution under $\tilde{H}_1$.
  \end{tablenotes}
\end{table}

To compute $T_n$ and reach the asymptotic normality, it is essential to select a suitable $\epsilon$.
To do so, \cite{Biometrika_Mikael_2022} considered a quadratic errors criterion to select the optimal $\epsilon$, which leads to favorable results.
However, we found this criterion appears to be ineffective in our test.
We also attempted other methods to determine the optimal $\epsilon$, including a double-bootstrap procedure from \cite{CJS_CaoVan_2006}, and a bootstrap method from \cite{JRSSB_MikaelIngrid_2019}.
We finally found that when the value of $\epsilon$ is close to $t_{(n)}$, the empirical level and empirical power perform well.
We show the detailed simulation results for different $\epsilon$ in the Supplementary Material.
It is worthwhile noting that when $\epsilon \in (t_{(n)}-\tilde{t}_{(K)}, 2(t_{(n)}-\tilde{t}_{(K)})]$, $T_n$ always equals $0$.
Thus, we select $\epsilon$ based on two cases: whether $t_{(n)}$ is greater than $2(t_{(n)}-\tilde{t}_{(K)})$ or not.
When $2(t_{(n)}-\tilde{t}_{(K)}) < t_{(n)}$, we take $\epsilon$ as $\epsilon^{\ast}= \frac{9}{8}t_{(n)}-\frac{1}{4}\tilde{t}_{(K)}$ and otherwise fix it to $t_{(n)}$.
Besides, if the denominator of $\hat{p}_G(\epsilon^{\ast})$ equals $0$ or $\hat{p}_G(\epsilon^{\ast})<\hat{p}_n$, we redefine $\hat{p}_G(\epsilon^{\ast})=\hat{p}_n$ and hence set $T_n=0$. In addition, $\hat{p}_G(\epsilon^{\ast})$ is redefined as $1$ if $\hat{p}_G(\epsilon^{\ast})>1$.

We evaluate the performance of the test through its empirical level and power, where the bootstrap critical value is obtained from $500$ bootstrap iterations.
Tables~\ref{Tab1}-\ref{Tab3} summarize the results for three considered distributions.
In Table~\ref{Tab1} under $H_0$, for a large $\lambda$, the empirical level is quite close to the nominal level, even for the small sample size $(n=400)$ and moderate censoring rate $(p=0.7)$.
Besides, for a small $\lambda$, the empirical level is slightly higher than the nominal level when the sample size is small and the censoring rate is lower.
For example, when $\lambda=1.5$, $n=400$ and $p=0.9$, the empirical level is $0.106$ while the empirical level quickly reduces to the nominal level when either the sample size or $\lambda$ increase.
This phenomenon is reasonable since, the smaller value of $\lambda$ corresponds to the larger censoring rate.
For example, when $\lambda=3$, the censoring rate is $27\%$ at $p=0.9$.
As $\lambda$ decreases to $1.5$, the censoring rate reaches $39\%$.
Furthermore, under $H_1$ we observe that for each setting, the empirical power increases to 1 as the sample size gets larger. For small sample sizes, however, the empirical power sometimes falls short of expectations. For example, when $\mu=3.5$, $p=0.7$ and $n=400$, the empirical power is $0.134$, which is far from $1$.
A larger sample size leads to a substantial increase in the empirical power.
Besides, the empirical power also improves with decreasing $\mu$ since smaller $\mu$ means larger departures from the null hypothesis.

In Tables~\ref{Tab2}-\ref{Tab3}, we observe that the empirical level is similar to that in Table~\ref{Tab1}.
In addition, for a high non-cure rate, the empirical powers in Tables~\ref{Tab2}-\ref{Tab3} also quickly tend to 1 as the sample size increases like those in Table~\ref{Tab1}.
However, for a moderate non-cure rate, compared to the empirical powers shown in Table~\ref{Tab1}, the empirical powers in Tables~\ref{Tab2}-\ref{Tab3} approach 1 at a slightly slower rate as the sample size increases.
The reason for this is that the censoring rate in Tables~\ref{Tab2}-\ref{Tab3} is larger than that in Table~\ref{Tab1} for each respective setting.

As previously mentioned in Section~\ref{Intro}, there are several tests in the statistical literature related to this testing problem.
However, their testing procedures are not suitable for comparison with our proposed test as they all considered insufficient follow-up as the null hypothesis, but the null hypothesis in this article is sufficient follow-up.
Nevertheless, we still show the simulation results of \cite{Book_MallerZhou_1996}'s testing procedure in Table~\ref{Tab4} to evaluate the performance of their test when the distribution of the failure time for susceptibles is a Weibull distribution.
To avoid ambiguity, in Table~\ref{Tab4}, we denote $\tilde{H}_0$ and $\tilde{H}_1$ as the null hypothesis of insufficient follow-up and the alternative hypothesis of sufficient follow-up, respectively.
For convenience, we denote their test statistic as $Q_n$.
Here, the empirical level and empirical power of their test are obtained based on the asymptotic distribution of $Q_n$, established in \cite{CJS_MallerEtal_2023}, where the parameter $\gamma$ in the asymptotic distribution of $Q_n$ is set to $1$.
The results in Table~\ref{Tab4} indicate an unsatisfactory performance of their test since all empirical levels are significantly greater than the nominal level, even in the case of large sample sizes.
In particular, under $\tilde{H}_0$, longer follow-up periods result in larger empirical levels.

\section{Real data analysis}\label{Real_data}

\subsection{Leukemia data}

Leukemia is the most prevalent form of cancer among children and ranks as the seventh primary cause of cancer-related mortality in the United States.
At present, the survival rates of leukemia patients are steadily increasing.
\cite{CA-SiegelEtal-2023} showed the $5$-year survival rates of leukemia patients for three distinct time periods: $1975$-$1977$, $1995$-$1997$, and $2012$-$2018$, which were determined to be $34\%$, $48\%$, and $66\%$, respectively.
The significant improvement in leukemia patients $5$-year survival rates can be attributed to the advancements in medical standards.
For example, the chimeric antigen receptor (CAR) T cells technology \citep{Science_JuneEtal_2018}, the most advanced technology available to treat leukemia, has made rapid advancements and received commercial approvals recently.
In the future, we can expect that more patients will never experience a relapse or recurrence of leukemia after treatment.

Here, we will analyze the leukemia data from the Surveillance, Epidemiology and End Results (SEER) Program of the National Cancer Institute (NCI) to illustrate the proposed testing procedure.
The SEER database collects cancer case data from population-based cancer registries that recorded approximately $47.9\%$ of the United States population, and has been utilized by numerous researchers, statisticians, and policy makers in the United States.
Besides, an annual update keeps the SEER database current.
The complete SEER database can be obtained from the website https://seer.cancer.gov/.
The information the SEER registers collected includes patient demographics such as gender and age, year of diagnosis, primary tumor site, stage at diagnosis, follow-up time and vital status, among others.
The leukemia data we will analyze are extracted from the database `Incidence-SEER Research Data, 18 Registries, Nov 2020 Sub (2000-2018)'.

The follow-up time is the duration from the diagnosis of leukemia to the death caused by leukemia; times to death due to other causes than leukemia or the end of follow-up are considered as censored.
We only consider the patients with acute monocytic leukemia, and exclude patients with zero or unknown survival time and those with unknown vital status recode.
After excluding the patients mentioned above, the cohort of interest consists of $2624$ patients out of which $1416$ are male and $1208$ are female.
The censoring rate is $35.4\%$, and the uncensored survival time ranges from $0.08$ to $12.50$ years, with a median of $0.58$ years.

We first present the KM curve in Fig~\ref{KM_Leu}.
The KM curve levels off at about $0.283$ and has a long plateau.
Besides, a total of $210$ patients have censored survival times exceeding the largest uncensored survival time.
Hence, it appears that the presence of a cure fraction is plausible and the assumption of sufficient follow-up holds.
Next, we present an illustrative analysis of the leukemia data to demonstrate the proposed testing procedure.
As in the simulations, we set $\epsilon$ to $\epsilon^{\ast}= \frac{9}{8}t_{(n)}-\frac{1}{4}\tilde{t}_{(K)}$ when $2(t_{(n)}-\tilde{t}_{(K)}) < t_{(n)}$; otherwise, we set it to $t_{(n)}$.
For the leukemia data, the largest observed time $t_{(n)}$  and the largest uncensored observed time $\tilde{t}_{(K)}$ are $18.92$ and $12.50$, respectively, and thus $2(t_{(n)}-\tilde{t}_{(K)}) < t_{(n)}$ and $\epsilon=18.16$.
The nominal level $\alpha$ is $0.05$ and the bootstrap critical value is obtained from $1000$ bootstrap iterations.
Calculating $\hat{p}_n$ and $\hat{p}_G(\epsilon^{\ast})$ as $0.7171$ and $0.7173$, respectively, our test statistic equals $0.0002$, which is very close to 0.
We further calculate the $p$-value to be $0.1420$.
Therefore, we fail to reject the null hypothesis.
Finally, as a comparison we also implement the test of \cite{Book_MallerZhou_1996}.
It is valuable to highlight that the null hypothesis in \cite{Book_MallerZhou_1996} is insufficient follow-up.
All other settings remain consistent with those in the simulations.
We calculate the test statistic $Q_n$ as $0.0095$.
Besides, the resulting $p$-value is $0.0006$, which is very close to $0$.
Thus, their results reject the hypothesis of insufficient follow-up, which is consistent with our conclusions.

\begin{figure}
  \centering
  \includegraphics[width=400pt]{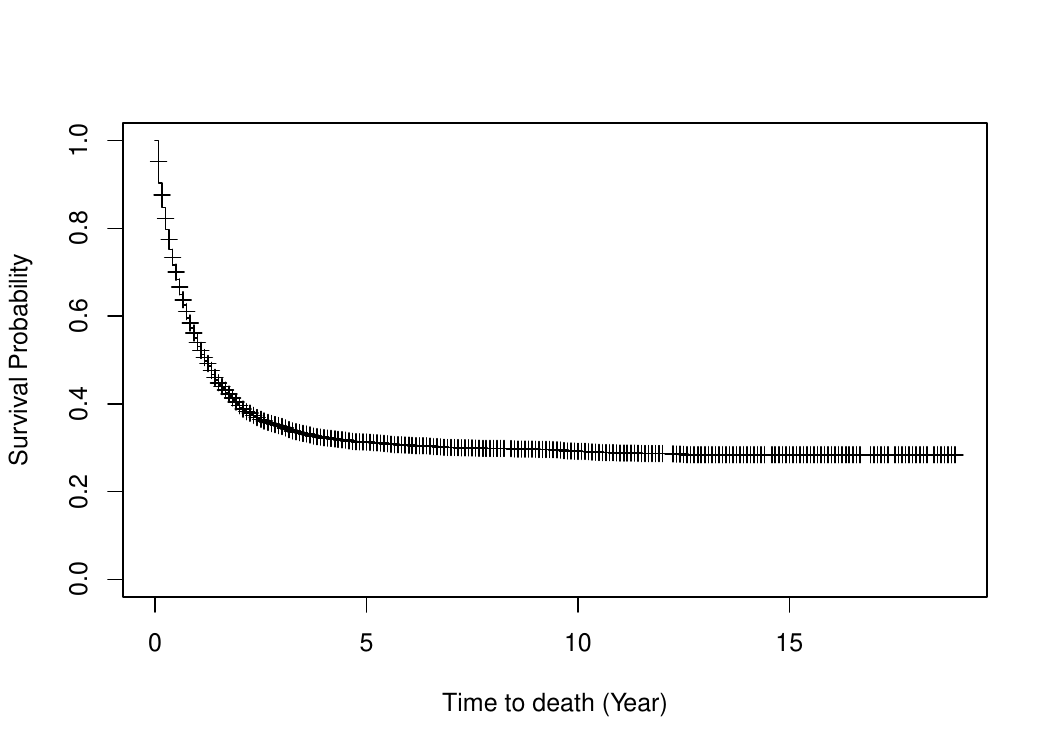}
  \caption{The Kaplan-Meier curve for the leukemia data}\label{KM_Leu}
\end{figure}

\subsection{Breast cancer data}

Female breast cancer, being one of the most prevalent types of cancer, has gained significant attention from statisticians, doctors and other researchers \citep{CE_HuntEtal_2014,JRSSC_ZhangEtal_2021}.
From 1989 to 2015, there was a notable reduction in breast cancer mortality by $39\%$ \citep{CA_DesantisEtal_2017}, thanks to significant advances in medical care.
For slowly proliferating tumours, a sufficient recording time is needed to estimate the survival rates.
\cite{BMC_TaiEtal_2005} analyzed the breast cancer data from the SEER database by assuming that the survival time of patients with breast cancer follows a log-normal distribution and determined the minimum required follow-up time for obtaining statistical cure rate estimates.
They showed that the minimum required follow-up time is 26.3 years for stage II of breast cancer whereas the actual records of the breast cancer data in the SEER database `Incidence-SEER Research Data, 18 Registries, Nov 2020 Sub (2000-2018)' are restricted, lasting only 19 years.
Here, we will apply the breast cancer data from the aforementioned SEER database to illustrate the proposed testing procedure.

We are interested in the time from the diagnosis of breast cancer to the death caused by breast cancer.
The censoring indicator is set to 1 if the patient dies from breast cancer, or $0$ if the patient is censored.
Here, we only consider the black females with breast cancer with tumour grade II and stage II.
Patients will be excluded if they have zero or unknown survival time.
Besides, we also exclude the patients with unknown cancer stage or unknown vital status recode.
After exclusions, there are $8618$ breast cancer patients in the cohort of interest.
The uncensored survival time ranges from $0.08$ to $18.25$ years, with a median of $4.33$ years, and the censoring rate is $81.3\%$.

The KM curve for the breast cancer data is presented in Fig~\ref{KM_breast}.
We observe that the KM curve has a short plateau.
An accurate assessment of the problem of testing for sufficient follow-up cannot be made through intuition alone.
Thus, we will apply the proposed testing procedure to check whether follow-up is sufficient for this breast cancer data.
We approximate the critical values in the test by employing the bootstrap procedure shown in Section \ref{Asym_theo} with $1000$ bootstrap iterations.
Here, the considered nominal level is $\alpha=0.05$.
As in the simulations, when $2(t_{(n)}-\tilde{t}_{(K)}) < t_{(n)}$, we take $\epsilon$ as $\epsilon^{\ast}= \frac{9}{8}t_{(n)}-\frac{1}{4}\tilde{t}_{(K)}$; otherwise, we fix it to $t_{(n)}$.
For the breast cancer data, $\epsilon$ is set to $16.72$ since $t_{(n)}=18.92$, $\tilde{t}_{(K)}=18.25$ and $2(t_{(n)}-\tilde{t}_{(K)}) < t_{(n)}$.
The estimators $\hat{p}_n$ and $\hat{p}_G(\epsilon^{\ast})$ are calculated as $0.3440$ and $0.4234$, which shows that the discrepancy between the naive estimator $\hat{p}_n$ and the estimator $\hat{p}_G(\epsilon^{\ast})$ is quite large.
The value of our test statistic is therefore $0.0794$.
We further calculate the $p$-value to be $0.0460$.
From our results, we reject the hypothesis of sufficient follow-up.
As a comparison, we further perform the test of \cite{Book_MallerZhou_1996}.
The remaining settings align with those used in the simulation.
The value of the test statistic $Q_n$ is $0.0002$ and $Q_n$ provides a $p$-value equal to $0.4219$.
Hence, the hypothesis of insufficient follow-up is not rejected from their results.
The results of both testing procedures suggest that the assumption of sufficient follow-up is not met.

\begin{figure}
  \centering
  \includegraphics[width=400pt]{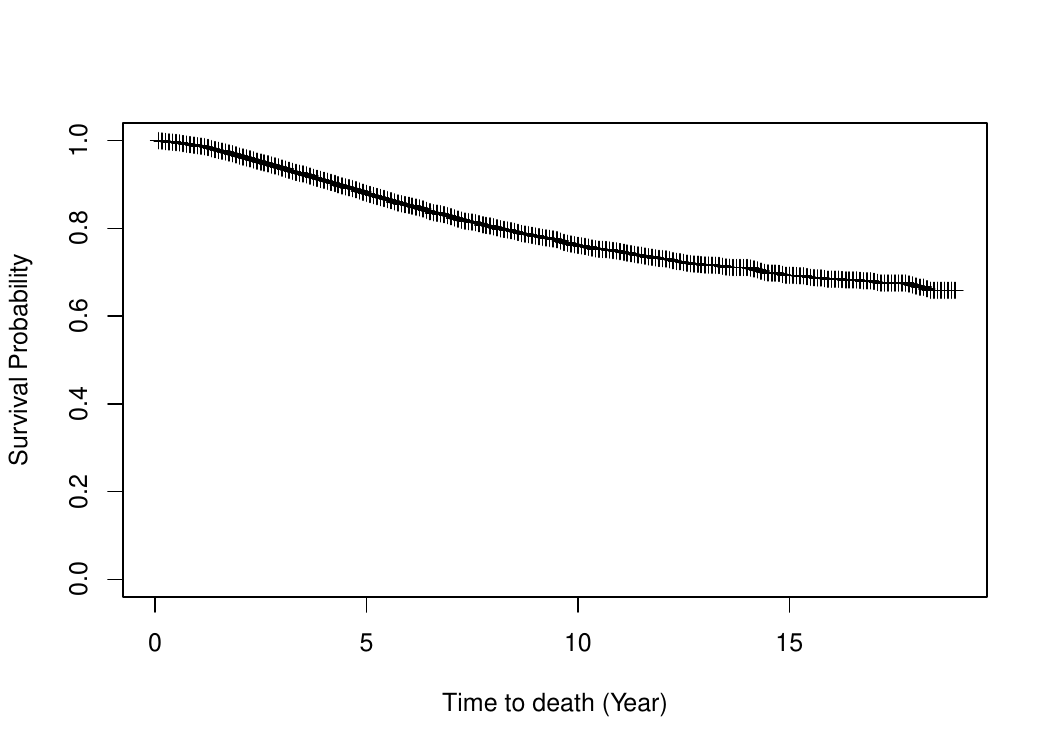}
  \caption{The Kaplan-Meier curve for the breast cancer data}\label{KM_breast}
\end{figure}

\section{Concluding remarks}\label{Diss}

In this article, we proposed a simple novel test to check the assumption of sufficient follow-up.
The proposed test statistic is based on the difference between two nonparametric estimators of the non-cure rate.
We established the asymptotic normality of our test statistic under the null hypothesis.
Extensive numerical studies have shown a satisfactory finite sample performance of our test.

We now point out some interesting future directions.
Our proposed testing procedure can be extended to the case when distributions are in the Frech\'{e}t domain of attraction or the Weibull domain of attraction.
The extension of our test is straightforward using our proposed approach by replacing the estimator $\hat{p}_{G}(\epsilon)$ in $T_n$ with the nonparametric estimator of the non-cure rate proposed by \cite{JRSSB_MikaelIngrid_2019}.
Besides, in our proposed testing procedure, we only considered the case without covariates.
In practice, however, covariates may be included in the model.
Hence, the objective of the second research is to extend the proposed testing procedure to the case involving covariates.
When covariates are incorporated into the model, \cite{CJS_XuPeng_2014} proposed a nonparametric estimator of the cure rate under sufficient follow-up, while \cite{CSDA_EscobaVanKeilegom_2023} further suggested another nonparametric estimator of the cure rate under insufficient follow-up.
Thus, the main idea is based on the difference between two nonparametric estimators of the cure rate proposed by \cite{CJS_XuPeng_2014} and \cite{CSDA_EscobaVanKeilegom_2023}, which gives us an indicator to check the assumption of sufficient follow-up in the presence of covariates.

\section*{Acknowledgment}

Ingrid Van Keilegom acknowledges support from the FWO and F.R.S.-FNRS under the Excellence of Science (EOS) programme, project ASTeRISK (grant No. 40007517).
The authors thank Xiaoguang Wang for helpful discussions and the Surveillance, Epidemiology, and End Results (SEER) Program of the National Cancer Institute for their commendable efforts in the creation of the SEER database.
The datasets that support the findings in this paper can be accessed from the website https://seer.cancer.gov/data/.
These datasets are used under license for this study and subject to restrictions on availability.

\section*{Supplementary material}

The R code and additional numerical simulations for different $\epsilon$ are available with this paper at Biometrika online.

\section*{Appendix}

\begin{proof}[Proof of Theorem~\ref{theorem2}]
To simplify the notations, we use $\hat{F}_{(d)}=\hat{F}_n(t_{(n)}-d)$ and $F_{(d)}=F(\tau_H-d)$, $d=0, \epsilon/2, \epsilon$.
Under sufficient follow-up and $0<p<1$, we have $\tau_H=\tau_{F_c}$.
First, we verify that
\begin{eqnarray}\label{minusepsilon}
\sqrt{n} (\hat{F}_{(\epsilon)}-F_{(\epsilon)})=X_n(t_{(n)}-\epsilon)+o_p(1),
\end{eqnarray}
where $X_n(t)$ is given in~\eqref{Xn_expression}.
To this end, we only need to prove $\sqrt{n} \{F(t_{(n)}-\epsilon)-F_{(\epsilon)}\}=o_p(1)$.
It follows by Taylor expansions that $\sqrt{n} \{F(t_{(n)}-\epsilon)-F_{(\epsilon)}\}=\sqrt{n}F'(\tau_H-\epsilon)(t_{(n)}-\tau_H)+o_p(\sqrt{n}|t_{(n)}-\tau_H|)$,
where $F'$ is the derivative of $F$.
For all $\delta>0$, we have
\begin{eqnarray*}
P(\sqrt{n}|t_{(n)}-\tau_H|>\delta)&=& P(t_{(n)}< \tau_H-\delta/\sqrt{n}) \\
&=& H^n(\tau_H-\delta/\sqrt{n})= \exp\{n \log H(\tau_H-\delta/\sqrt{n})\}.
\end{eqnarray*}
By Assumption~3 and $\tau_H=\tau_{F_c}$, we have
\begin{eqnarray*}
\lim_{n \rightarrow \infty}n \log H(\tau_H-\delta/\sqrt{n})&=&
\lim_{n \rightarrow \infty} n\{H(\tau_H-\delta/\sqrt{n})-1\}\\
&=& -\lim_{n \rightarrow \infty} n\{\bar{F}(\tau_H-\delta/\sqrt{n})\}\{\bar{F}_c(\tau_H-\delta/\sqrt{n})\}=-\infty.
\end{eqnarray*}
This shows that $P(\sqrt{n}|t_{(n)}-\tau_H|>\delta)\rightarrow 0$ in probability as $n \rightarrow \infty$.
Thus \eqref{minusepsilon} holds as claimed.
Following similar arguments, we can prove that
\begin{eqnarray*}
\sqrt{n} (\hat{F}_{(\epsilon/2)}-F_{(\epsilon/2)})=X_n(t_{(n)}-\epsilon/2)+o_p(1)
\end{eqnarray*}
and
\begin{eqnarray*}
\sqrt{n} (\hat{F}_{(0)}-F_{(0)})=X_n(t_{(n)})+o_p(1).
\end{eqnarray*}
The proof of Theorem~4.3 from \cite{Book_MallerZhou_1996} showed that the conditions of Theorem~\ref{theorem1} are satisfied under the conditions of Theorem~\ref{theorem2}.
Therefore, Theorem~1 holds.
Then, by Theorem~\ref{theorem1}, $X_n(t)$ converges weakly in $D[0,\tau_H]$ to $X(t)$ as $n \rightarrow \infty$.
In addition, $t_{(n)}$ converges in probability to $\tau_H$ as $n \rightarrow \infty$ \citep[Theorem 3.2]{Book_MallerZhou_1996}.
Based on Theorem~3.9 of \cite{Book_Billingsley_1999}, we have $(X_n(t), t_{(n)})$ converges weakly in $D[0,\tau_H]\times [0,\tau_H]$ to $(X(t),\tau_H)$ as $n \rightarrow \infty$.
Using Lemma on page~$151$ of \cite{Book_Billingsley_1999}, we prove that $X_n(t_{(n)})$ converges in distribution to $X(\tau_H)$ as $n\rightarrow \infty$.
Following similar arguments, we can prove that $X_n(t_{(n)}-\epsilon)$ and $X_n(t_{(n)}-\epsilon/2)$ converge in distribution to $X(\tau_H-\epsilon)$ and $X(\tau_H-\epsilon/2)$, respectively, as $n\rightarrow \infty$.

Next, we decompose $\sqrt{n} (T_n-b_{\epsilon})$ as follows:
\begin{eqnarray}\label{decompose}
\sqrt{n}(T_n -b_{\epsilon}) &=& \sqrt{n} \left(\frac{\hat{F}^2_{(\epsilon/2)}-\hat{F}_{(0)}\hat{F}_{(\epsilon)}}{2\hat{F}_{(\epsilon/2)}-\hat{F}_{(0)}-\hat{F}_{(\epsilon)}}
-\frac{F^2_{(\epsilon/2)}-F_{(0)}F_{(\epsilon)}}{2F_{(\epsilon/2)}-F_{(0)}-F_{(\epsilon)}}\right) 
 -\sqrt{n}(\hat{F}_{(0)}-F_{(0)}) \notag \\
&=& \sqrt{n}\left(\frac{1}{2\hat{F}_{(\epsilon/2)}-\hat{F}_{(0)}-\hat{F}_{(\epsilon)}}
-\frac{1}{2F_{(\epsilon/2)}-F_{(0)}-F_{(\epsilon)}}\right)
 (\hat{F}^2_{(\epsilon/2)}-\hat{F}_{(0)}\hat{F}_{(\epsilon)}) \nonumber \\
&&+\sqrt{n}\frac{\hat{F}^2_{(\epsilon/2)}-\hat{F}_{(0)}\hat{F}_{(\epsilon)}+F_{(0)}F_{(\epsilon)}-F^2_{(\epsilon/2)}}{2F_{(\epsilon/2)}-F_{(0)}-F_{(\epsilon)}} -\sqrt{n}(\hat{F}_{(0)}-F_{(0)}).
\end{eqnarray}
By the continuous mapping theorem and Slutsky's lemma, the first term in~\eqref{decompose} can be shown to be
\begin{eqnarray*}
\{X_n(t_{(n)}-\epsilon) -2X_n(t_{(n)}-\epsilon/2)+X_n(t_{(n)})\}\frac{F^2_{(\epsilon/2)}-F_{(0)}F_{(\epsilon)}}{(2F_{(\epsilon/2)}-F_{(\epsilon)}-F_{(0)})^2}+o_p(1).
\end{eqnarray*}
Likewise, the second term in~\eqref{decompose} equals
\begin{eqnarray*}
\{2F_{(\epsilon/2)}X_n(t_{(n)}-\epsilon/2)-F_{(0)}X_n(t_{(n)}-\epsilon)-F_{(\epsilon)}X_n(t_{(n)})\}
\frac{1}{2F_{(\epsilon/2)}-F_{(\epsilon)}-F_{(0)}}+o_p(1).
\end{eqnarray*}
It follows from the above analysis that
\begin{eqnarray*}
\sqrt{n}(T_n -b_{\epsilon})=s_1X_n(t_{(n)}-\epsilon)+s_2X_n(t_{(n)}-\epsilon/2)+s_3X_n(t_{(n)})+o_p(1),
\end{eqnarray*}
where
\begin{eqnarray}\label{s1}
s_1=\frac{(F_{(\epsilon/2)}-F_{(0)})^2}{(2F_{(\epsilon/2)}-F_{(\epsilon)}-F_{(0)})^2},
\end{eqnarray}
\begin{eqnarray}\label{s2}
s_2=\frac{2(F_{(0)}-F_{(\epsilon/2)})(F_{(\epsilon)}-F_{(\epsilon/2)})}{(2F_{(\epsilon/2)}-F_{(\epsilon)}-F_{(0)})^2}
\end{eqnarray}
and
\begin{eqnarray}\label{s3}
s_3=\frac{(F_{(\epsilon/2)}-F_{(\epsilon)})^2}{(2F_{(\epsilon/2)}-F_{(\epsilon)}-F_{(0)})^2}-1.
\end{eqnarray}

Finally, we have that $\sqrt{n}(T_n -b_{\epsilon})$ converges in distribution to a normal distribution with zero mean and variance
\begin{eqnarray*}
\sigma^2_{\epsilon}=\sum_{i=1}^3\sum_{j=1}^3 s_i s_j \bar{F}\{\tau_H+\epsilon(i-3)/2\}
\bar{F}\{\tau_H+\epsilon(j-3)/2\} v\{\tau_H+\epsilon(i\wedge j-3)/2\}.
\end{eqnarray*}
This finishes the proof.
\end{proof}

\bibliography{mybibfile}
\end{document}